\documentclass[conference,a4paper]{IEEEtran}

\usepackage[left=1.46cm,right=1.46cm,top=1.5cm,bottom=2cm]{geometry}

\usepackage[utf8]{inputenc} 
\usepackage[T1]{fontenc}
\usepackage{url}
\usepackage{ifthen}
\usepackage{cite}
\usepackage[cmex10]{amsmath} % Use the [cmex10] option to ensure complicance
\allowdisplaybreaks[3]  % split equations

\usepackage{amssymb,amsthm}

\newtheorem{Theorem}{Theorem}
\newtheorem{Lemma}{Lemma}
\newtheorem{Proposition}{Proposition}

\newtheorem{Definition}{Definition}

\usepackage{microtype} % see \textls

\ifCLASSINFOpdf
\else
\fi

\newcommand{\E}{\mathbb{E}}
\newcommand{\R}{\mathbb{R}}
\renewcommand{\d}{\,\mathrm{d}}
\let\phi\varphi
\let\tilde\widetilde

\begin{document}
%\title{\centerline{\relscale{.93}Transportation Proofs of R\'enyi Entropy Power Inequalities}} 
\title{R\'enyi Entropy Power and Normal Transport}

 %%% Single author, or several authors with same affiliation:
 \author{%
   \IEEEauthorblockN{Olivier Rioul}
   \IEEEauthorblockA{LTCI, T\'el\'ecom Paris, Institut Polytechnique de Paris, 91120, Palaiseau, France}
%   \\[-3.2ex]
 }

\maketitle

\begin{abstract}
A framework for deriving Rényi entropy-power inequalities (REPIs) is presented that uses linearization and an inequality of Dembo, Cover, and Thomas. Simple arguments are given to recover the previously known Rényi EPIs and derive new ones, by unifying a multiplicative form with constant c and a modification with exponent $\alpha$ of previous works. An information-theoretic proof of the Dembo-Cover-Thomas inequality---equivalent to Young's convolutional inequality with optimal constants---is provided, based on properties of Rényi conditional and relative entropies and using transportation arguments from Gaussian densities. For log-concave densities, a transportation proof of a sharp varentropy bound is presented.
\end{abstract}

{\small
This work was partially presented at the 2019 Information Theory and Applications Workshop, San Diego, CA.}

\section{Introduction}

We consider the $r$-entropy ({R\'enyi entropy} of exponent $r$, where $r>0$ and $r\ne 1$) of a $n$-dimensional zero-mean random vector $X\in\R^n$ having density $f\in L^r(\R^n)$:\\[-1.5ex]
\begin{equation}
h_r(X)=\frac{1}{1-r}\log \int_{\R^n} f^r(x) \d x 
= -r' \log\|f\|_r \label{hpdef2}
\end{equation}
where $\|f\|_r$ denotes the $L^r$ norm of $f$, and $r'=\frac{r}{r-1}$ is the \emph{conjugate exponent} of $r$, such that $\frac{1}{r}+\frac{1}{r'}=1$.
Notice that either $r>1$ and $r'>1$, or $0<r<1$ and $r'<0$.
The limit as $r\to 1$ is the classical
$h_1(X)=h(X)=  -\int_{\R^n} f(x) \log f(x) \d x$. Letting $N(X)=\exp\bigl(2h(X)/n\bigr)$ be the corresponding entropy power~\cite{Shannon48}, 
the famous entropy power inequality (EPI)~\cite{Shannon48,Rioul11} writes
$N\Bigl( \sum_{i=1}^m X_i\Bigr) \geq   \sum_{i=1}^m N(X_i)$ for any independent random vectors $X_1,X_2,\ldots,X_m\in\R^n$.
The link with the Rényi entropy $h_r(X)$ was first made in~\cite{DemboCoverThomas91} in connection with a strengthened Young's convolutional inequality, where the EPI is obtained by letting exponents tend to $1$~\cite[Thm~17.8.3]{CoverThomas06}.

Recently, there has been increasing interest in {R\'enyi} entropy-power inequalities~\cite{MadimanMelbourneXu17}.
The Rényi entropy-power $N_r(X)$ is defined~\cite{BobkovChistyakov15} as the average power of a white Gaussian vector having the same Rényi entropy as~$X$. 
%Write $X^*\sim \mathcal{N}(0,\mathbf{K})$ if $X^*$ is normally distributed with $n\times n$ covariance matrix $\mathbf{K}$.
If $X^*\sim \mathcal{N}(0,\sigma^2\mathbf{I})$ is white Gaussian, an easy calculation yields
\begin{equation}\label{renyiGauss}
h_r(X^*) = \tfrac{n}{2} \log (2\pi\sigma^2) + \tfrac{n}{2} r' \tfrac{\log r}{r}. 
\end{equation}
%$X^*\sim \mathcal{N}(0,\mathbf{K})$ is Gaussian with nonsingular covariance matrix $\mathbf{K}$, an easy calculation yields
%\begin{equation}
%h_p(X^*) = \frac{1}{2} \log ((2\pi)^n|\mathbf{K}|) + \frac{n}{2} p' \frac{\log p}{p}
%\end{equation}
%where $|\cdot|$ denotes the determinant. 
Since equating $h_r(X^*)=h_r(X)$ gives $\sigma^2 =  \frac{e^{2h_r(X)/n}}{2\pi r^{r'/r}}$, we define $N_r(X)= e^{2h_r(X)/n}$ as the $r$-entropy power.
.

Bobkov and Chistyakov~\cite{BobkovChistyakov15} extended the classical EPI to the $r$-entropy by incorporating a $r$-dependent constant $c>0$:
\begin{equation}\label{repic}
N_r\Bigl( \sum\nolimits_{i=1}^m X_i\Bigr) \geq c  \sum\nolimits_{i=1}^m N_r(X_i).
\end{equation}
Ram and Sason~\cite{RamSason16} improved (increased) the value of $c$ by making it depend also on the number $m$ of independent vectors $X_1,X_2,\ldots,X_m$. 
Bobkov and Marsiglietti~\cite{BobkovMarsiglietti17} proved another modification of the EPI for the Rényi entropy:
\begin{equation}\label{repialpha}
{N_r^{\vphantom{2}}}^{\!\alpha}\Bigl( \sum\nolimits_{i=1}^m X_i\Bigr) \geq  \sum\nolimits_{i=1}^m {N_r^{\vphantom{2}}}^{\!\alpha}(X_i)
\end{equation}
with~a power exponent parameter $\alpha>0$. Due to the non-increasing property of the $\alpha$-norm, if~\eqref{repialpha} holds for $\alpha$ it also holds for any $\alpha'>\alpha$.
The value of $\alpha$ was further improved (decreased) by Li~\cite{Li18}. 
All the above EPIs were found for Rényi entropies of orders $r>$1. 
Recently, the $\alpha$-modification of the Rényi EPI~\eqref{repialpha} was extended to orders~$<$1 for two independent variables having log-concave densities by Marsiglietti and Melbourne~\cite{MarsigliettiMelbourne18}. 
The starting point of all the above works was Young's strengthened convolutional~inequality. 

In this paper, we build on the results of~\cite{Rioul18} to provide simple proofs for Rényi EPIs of the general form 
\begin{equation}\label{repig}
{N_r^{\vphantom{2}}}^{\!\alpha}\Bigl( \sum\nolimits_{i=1}^m X_i\Bigr) \geq  c\sum\nolimits_{i=1}^m {N_r^{\vphantom{2}}}^{\!\alpha}(X_i)
\end{equation}
with constant $c>0$ and exponent $\alpha>0$.
The present framework uses only basic properties of Rényi entropies and
 is based on a transportation argument from normal densities and a change of variable by rotation, which was previously used to give a simple proof of Shannon's original EPI~\cite{Rioul17}.

\section{Linearization}

\noindent The first step toward proving~\eqref{repig} is the following linearization lemma which generalizes~\cite[Lemma~2.1]{Li18}.
\begin{Lemma}\label{charact}
For independent $X_1,X_2,\ldots,X_m$, the Rényi EPI in the general form~\eqref{repig} is equivalent to the following inequality
\begin{equation}\label{repi1bis}
h_r\bigl( \textstyle\sum\limits_{i=1}^m\sqrt{\lambda_i}X_i\bigr) - \sum\limits_{i=1}^m \lambda_i h_{r}(X_i)
\geq \frac{n}{2} \bigl(  \frac{\log c}{\alpha} + \bigl(\frac{1}{\alpha}-1\bigr) H(\lambda)  \bigr)
\end{equation}
for any distribution  $\lambda=(\lambda_1,\ldots,\lambda_m)$ of entropy $H(\lambda)$.
\end{Lemma}

\begin{proof}
Note the scaling property $h_r(aX)=h_r(X) + n\log|a|$ for any $a\ne 0$, established by a change of variable. It follows that 
%the R\'enyi entropy power enjoys the same scaling property as for the usual power: 
$N_r(aX)=a^2 N_r(X)$.
Now first suppose~\eqref{repig} holds. Then
\begin{align}
h_r\bigl( &\textstyle\sum_{i=1}^m\sqrt{\lambda_i}X_i\bigr)  =\tfrac{n}{2\alpha}\log {N_r^{\vphantom{2}}}^\alpha\bigl( \textstyle\sum_{i=1}^m\sqrt{\lambda_i}X_i\bigr)\\
&\geq \tfrac{n}{2\alpha}\log\textstyle\sum_{i=1}^m {N_r^{\vphantom{2}}}^\alpha(\sqrt{\lambda_i}X_i) + \tfrac{n}{2\alpha}\log c\notag\\
&=\tfrac{n}{2\alpha}\log\textstyle\sum_{i=1}^m \lambda_i^\alpha {N_r^{\vphantom{2}}}^\alpha(X_i) + \tfrac{n}{2\alpha}\log c   \label{a}\\
&\geq \tfrac{n}{2\alpha}\textstyle\sum_{i=1}^m \lambda_i \log\bigl( \lambda_i^{\alpha-1} {N_r^{\vphantom{2}}}^\alpha(X_i)\bigr) + \tfrac{n}{2\alpha}\log c  
\label{b}\\
&= \textstyle\sum_{i=1}^m \lambda_i h_{r}(X_i) +\tfrac{n(\alpha-1)}{2\alpha}\textstyle\sum_{i=1}^m \lambda_i\log {\lambda_i}
%\notag\\&\quad
+ \tfrac{n}{2\alpha}\log c   \notag
\end{align}
which proves~\eqref{repi1bis}.
The scaling property is used in~\eqref{a} and the concavity of the logarithm is used in~\eqref{b}. 

Conversely, suppose that~\eqref{repi1bis} is satisfied for all $\lambda_i>0$ such that $\sum_{i=1}^m \lambda_i=1$. Set~$\lambda_i~=~{N_r^{\vphantom{2}}}^\alpha(X_i)/ \sum_{i=1}^m {N_r^{\vphantom{2}}}^\alpha(X_i)$.
 Then
\begin{align*}
{N_r^{\vphantom{2}}}^\alpha\bigl( &\textstyle\sum_{i=1}^m X_i\bigr) 
= \exp \tfrac{2\alpha}{n} h_r\bigl( \textstyle\sum_{i=1}^m\sqrt{\lambda_i}\frac{X_i}{\sqrt{\lambda_i}}\bigr)\\
&\geq \exp \tfrac{2\alpha}{n} \textstyle\sum_{i=1}^m \lambda_i h_{r}\Bigl(\frac{X_i}{\sqrt{\lambda_i}}\Bigr) %\notag\\&\qquad 
\times  c \!\cdot\! e^{(1-\alpha) \textstyle\sum_{i=1}^m \lambda_i\log \frac{1}{\lambda_i}}\\
&= c \textstyle\prod\limits_{i=1}^m  \Bigl({N_r^{\vphantom{2}}}^\alpha\Bigl(\frac{X_i}{\sqrt{\lambda_i}}\Bigr) \lambda_i^{\alpha-1} \Bigr)^{\lambda_i} 
%\\&
= c \textstyle\prod\limits_{i=1}^m  \Bigl({N_r^{\vphantom{2}}}^\alpha(X_i)
\lambda_i^{-1} \Bigr)^{\lambda_i} 
\\&
=c\bigl(\textstyle\sum_{i=1}^m {N_r^{\vphantom{2}}}^\alpha(X_i)\bigr)^{\textstyle\sum_{i=1}^m \lambda_i}
%\\&
=c\textstyle\sum_{i=1}^m {N_r^{\vphantom{2}}}^\alpha(X_i).
\end{align*}
which proves~\eqref{repig}.
\end{proof}

\section{The REPI of Dembo-Cover-Thomas}

As a second ingredient we have the following result, which was essentially established by Dembo, Cover and Thomas~\cite{DemboCoverThomas91}. 
It is this Rényi version of the EPI which led them to prove Shannon's original EPI by letting Rényi exponents $\to 1$.

\begin{Theorem}\label{repi1m}
Let $r_1,\ldots,r_m,r$ be exponents those conjugates $r'_1,\ldots,r'_m,r'$ are of the same sign and satisfy
$
%\frac{1}{r'_1}+\frac{1}{r'_2}+\cdots+\frac{1}{r'_m}=\frac{1}{r'} 
\sum_{i=1}^m\frac{1}{r'_i}=\frac{1}{r'}
$
and let $\lambda_1,\ldots,\lambda_m$ be the discrete probability distribution 
%defined by
$\lambda_i =  \frac{r'}{r'_i}$. %($i=1,2,\ldots,m$).
%!TEX encoding = UTF-8 Unicode
Then, for independent zero-mean $X_1,X_2,\ldots,X_m$,
\begin{equation}
\begin{split}\label{repi1mineqGauss}
 h_r\Bigl( \textstyle\sum\limits_{i=1}^m\sqrt{\lambda_i}&X_i\Bigr) - \textstyle\sum\limits_{i=1}^m \lambda_i h_{r_i}(X_i)
\\&
\geq
h_r\Bigl(\textstyle\sum\limits_{i=1}^m\sqrt{\lambda_i}X^*_i\Bigr) -\textstyle\sum\limits_{i=1}^m \lambda_i h_{r_i}(X^*_i)
\end{split} 
\end{equation}
where $X^*_1,X^*_2,\ldots,X^*_m$ are i.i.d.\@ standard Gaussian $\mathcal{N}(0,\mathbf{I})$. 
Equality holds if and only if the $X_i$ are i.i.d.\@ {Gaussian}.
\end{Theorem}

It is easily seen from the expression~\eqref{renyiGauss} of the Rényi entropy of a Gaussian that \eqref{repi1mineqGauss} is equivalent to
%\begin{equation}
%\begin{split}
%h_r\bigl( \sqrt{\lambda_1}X_1+\sqrt{\lambda_2}X_2+\cdots+\sqrt{\lambda_m}X_m\bigr) 
%-\lambda_1 h_{r_1}(X_1) -\lambda_2 h_{r_2}(X_2) -\cdots - -\lambda_m h_{r_m}(X_m)
%\\\geq \frac{n}{2} r' \Bigl( \frac{\log r}{r}-\frac{\log r_1}{r_1}-\frac{\log r_2}{r_2}-\cdots-\frac{\log r_m}{r_m}\Bigr)
% \end{split}
%\end{equation}
\begin{equation}\label{repi1mineq}
h_r\Bigl( \textstyle\sum\limits_{i=1}^m\sqrt{\lambda_i}X_i\Bigr) - \textstyle\sum\limits_{i=1}^m \lambda_i h_{r_i}(X_i)
\geq \tfrac{n}{2} r' \Bigl( \frac{\log r}{r}-\textstyle\sum\limits_{i=1}^m\frac{\log r_i}{r_i}\Bigr)
\end{equation}
Note that the l.h.s. is very similar to that of~\eqref{repi1bis} except that different Rényi exponents are present.
This will be the crucial step toward proving~\eqref{repig}.

Theorem~\ref{repi1m} (for $m=2$) was derived in~\cite{DemboCoverThomas91} as a rewriting of Young's strengthened convolutional~inequality with optimal constants. Section~\ref{repi1sec} provides a simple transportation proof, which uses only basic properties of Rényi entropies.

\section{REPIs for Orders >1}\label{repi2sec}

If $r>1$, then $r'>0$ and all $r'_i$ are positive and greater than~$r'$. Therefore, all $r_i$ are less than~$r$. Using the well-known fact that $h_r(X)$ is non increasing in $r$ 
(see also~\eqref{identitydiff} below),
%% ????????
%%
\begin{equation}\label{monori}
h_{r_i}(X_i)\geq h_r(X_i) \qquad (i=1,2,\ldots,m).  
\end{equation}
Plugging this into~\eqref{repi1mineq}, one obtains
\begin{equation}\label{repi2mineq}
h_r\bigl( \textstyle\sum\limits_{i=1}^m\sqrt{\lambda_i}X_i\bigr) -\textstyle\sum\limits_{i=1}^m \lambda_i h_{r}(X_i)
\geq \frac{n}{2} r' \bigl( \frac{\log r}{r}-\sum_{i=1}^m\frac{\log r_i}{r_i}\bigr)
\end{equation}
where $\lambda_i=r'/r'_i$. % for $i=1,2,\ldots,m$. 
From Lemma~\ref{charact} is suffices to establish that the r.h.s. of this inequality exceeds that of~\eqref{repi1bis} to prove~\eqref{repig} for appropriate constants $c$ and $\alpha$.
For future reference define
%\footnote{The absolute value $|r'|$ is needed in the next section where $r'$ is negative.}
\begin{align}
A(\lambda)&=|r'| \bigl( \tfrac{\log r}{r}-\textstyle\sum\limits_{i=1}^m\tfrac{\log r_i}{r_i}\bigr)\\[-1.5ex]
&=|r'| %\Bigl( %\frac{\log r}{r}+ 
\sum_{i=1}^m (1-\tfrac{\lambda_i}{r'})\log(1\!-\!\tfrac{\lambda_i}{r'})
\!-\!(1\!-\!\tfrac{1}{r'})\log(1\!-\!\tfrac{1}{r'})
%\Bigr)
.\notag
\label{Alambda}
\end{align}
(The absolute value $|r'|$ is needed in the next section where $r'$ is negative.)
%%%
%%%HACK!!!!
%%%
This function is strictly convex in $\lambda=(\lambda_1,\lambda_2,\ldots,\lambda_m)$ because 
$x\mapsto (1-x/r')\log(1-x/r')$ is strictly convex. Note that $A(\lambda)$ vanishes in the limiting cases where $\lambda$ tends to one of the standard unit vectors $(1,0,\ldots,0)$, \ldots, $(0,0,\ldots,0,1)$ and since every $\lambda$ is a convex combination of these vectors and $A(\lambda)$ is strictly convex, one has $A(\lambda)<0$.

Using the properties of $A(\lambda)$ it is immediate to recover known Rényi EPIs:

\begin{Proposition}[Ram and Sason~\cite{RamSason16}]\label{thmrepic}
The Rényi EPI~\eqref{repic} holds for $r>1$ and $c=r^{r'/r}\bigl(1-\frac{1}{mr'}\bigr)^{mr'-1}$.
\end{Proposition}

\begin{proof}
By Lemma~\ref{charact} for $\alpha=1$ we only need to check that the r.h.s. of~\eqref{repi2mineq} is greater than $\frac{n}{2}\log c$ for any choice of the $\lambda_i$'s, that is, for any choice of exponents $r_i$ such that $\sum_{i=1}^m \frac{1}{r'_i}~=~\frac{1}{r'}$. Thus, \eqref{repic} will hold for 
$\log c = \min_{\lambda} A(\lambda)$. Now, by the log-sum inequality~\cite[Thm~2.7.1]{CoverThomas06},
\begin{align}
\sum\limits_{i=1}^m\frac{1}{r_i}\log \frac{1}{r_i} &\!\geq\!  \bigl(\sum\limits_{i=1}^m\frac{1}{r_i}\bigr) \log \frac{\textstyle\sum_{i=1}^m\frac{1}{r_i}}{m} %\\&
= (m-\tfrac{1}{r'})\log \frac{m-\frac{1}{r'}}{m}
%\\[-5ex]\notag %%%% HACK!!!!! %%%%%
\end{align}
with equality if and only if all $r_i$ are equal, that is, the $\lambda_i$ are equal to $1/m$. Thus, 
$\min_\lambda A(\lambda)= r'  \bigl[ \frac{\log r}{r}+ (m-1/r')\log \frac{m-1/r'}{m} \bigr]=\log c$.
%which yields $c=r^{r'/r}\bigl(1-\frac{1}{mr'}\bigr)^{mr'-1}$.
\end{proof}
%An alternate proof is to argue that $A(\lambda)$ is convex and symmetrical in $\lambda=(\lambda_1,\lambda_2,\ldots,\lambda_m)$ and is, therefore, minimized when all $\lambda_i$ are equal.

Note that $\log c= r' \frac{\log r}{r} +(mr'-1) \log \bigl(1-\frac{1}{mr'}\bigr)<0$ decreases (and tends to $r' \frac{\log r}{r} -1$) as $m$ increases; in fact $\frac{\partial \log c}{\partial m} = r' \log \bigl(1-\frac{1}{mr'}\bigr) +\frac{mr'}{r'm^2} < r'(-\frac{1}{mr'})+\frac{1}{m}=0$. 
Thus, a universal constant independent of $m$ is obtained by taking%\vspace*{-1ex}
\begin{align}
c&= \inf_m \;r^{r'/r}\bigl(1-\frac{1}{mr'}\bigr)^{mr'-1} 
%\\&= r^{r'/r}\lim_{m\to\infty}\bigl(1-\frac{1}{mr'}\bigr)^{mr'-1}\\&
=\frac{r^{r'/r}}{e},
\end{align}
as was established by Bobkov and Chistyakov~\cite{BobkovChistyakov15}.

%The constant $c$ in Theorem~\ref{thmrepic} is not optimal because equality in~\eqref{monori} holds  only if the $X_i$ are uniformly distributed while equality in~\eqref{repi1mineq} holds if and only if the $X_i$ are normally distributed. Ram and Sason~\cite{RamSason16} tightened~\eqref{repic} further using optimization techniques, resulting in a constant that depends on the relative values of the entropy powers themselves. 

\begin{Proposition}[Li~\cite{Li18}]\label{thmrepialpha}
The Rényi EPI~\eqref{repialpha} holds for $r>1$ and 
$\alpha=\bigl[1+{r' \frac{\log_2 r}{r} +(2r'-1) \log_2 \bigl(1-\frac{1}{2r'}\bigr)}  \bigr]^{-1}$.
%$\alpha=\Bigl(1+\frac{r' \frac{\log r}{r} +(mr'-1) \log \bigl(1-\frac{1}{mr'}\bigr)}{\log m}  \Bigr)^{-1}$.
\end{Proposition}
%The case $m=2$ yields $\alpha=\frac{r-1}{(r+1)\log_2(r+1)-r\log_2 r -2}$ which was found by 
Li~\cite{Li18} remarked that this value of $\alpha$ is strictly smaller (better) than the value $\alpha=\frac{r+1}{2}$ obtained  previously by Bobkov and Marsiglietti~\cite{BobkovMarsiglietti17}.
In~\cite{Rioul18} it is shown that it cannot be further improved in our framework by making it depend on~$m$.

\begin{proof}
Since the announced $\alpha$ does not depend on $m$, we can always assume that $m=2$. By Lemma~\ref{charact} for $c=1$, we only need to check that the r.h.s. of~\eqref{repi2mineq} is greater than $\frac{n}{2}(1/\alpha-1)H(\lambda)$ for any choice of  $\lambda_i$s, that is, for any choice of exponents $r_i$ such that $\sum_{i=1}^2 \frac{1}{r'_i} = \frac{1}{r'}$. Thus,~\eqref{repialpha} will hold for 
$\frac{1}{\alpha}-1 = \min_\lambda \frac{A(\lambda)}{H(\lambda)}$.
Li~\cite{Li18} showed---this is also easily proved using~\cite[Lemma~8]{MarsigliettiMelbourne18}---that  the~minimum is obtained when $\lambda=(1/2,1/2)$. The corresponding value of  $A(\lambda)/H(\lambda)$ is $\bigl[r' \frac{\log r}{r} +(2r'-1) \log \bigl(1-\frac{1}{2r'}\bigr)\bigr]/\log 2=1/\alpha -1$.
\end{proof}

The above value of $\alpha$ is $>1$. However, using the same method, it is easy to obtain Rényi EPIs with exponent values $\alpha<1$. In this way we obtain a new Rényi EPI:
\begin{Proposition}\label{thmrepig}
The Rényi EPI~\eqref{repig} holds for $r>1$, $0<\alpha<1$ with $c=\bigl[m\;r^{r'/r}\bigl(1-\frac{1}{mr'}\bigr)^{mr'-1}\bigr]^\alpha / m$.
\end{Proposition}
\begin{proof}
{By Lemma~\ref{charact} we only need to check that the r.h.s. of~ Equation \eqref{repi2mineq} is greater than $\frac{n}{2}\bigl((\log c)/\alpha+(1/\alpha-1)H(\lambda)\bigr)$, that is, $A(\lambda)\geq (\log c)/\alpha+(1/\alpha-1)H(\lambda)$ for any choice of  $\lambda_i$s, that is, for any choice of exponents $r_i$ such that $\sum_{i=1}^m \frac{1}{r'_i} = \frac{1}{r'}$. Thus, for a given $0<\alpha<1$,  \eqref{repig} will hold for $\log c = \min_\lambda \alpha A(\lambda) - (1-\alpha) H(\lambda)$. From the preceding proofs (since both $A(\lambda)$ and $-H(\lambda)$ are convex functions of $\lambda$), the minimum is attained when all $\lambda_i$s are equal. This gives $\log c =\alpha \Bigl(r' \frac{\log r}{r} +(mr'-1) \log \bigl(1-\frac{1}{mr'}\bigr)\Bigr)-(1-\alpha) \log m$.}
\end{proof}

%$$
% \frac{\log c}{\alpha} + \bigl(\frac{1}{\alpha}-1\bigr) H(\lambda)  \leq  A(\lambda)
%$$
%$$
%\log c\leq \alpha A(\lambda) + \bigl(\alpha-1\bigr) H(\lambda)
%$$
 
\section{REPIs for Orders <1 and Log-Concave Densities}

If $r<1$, then $r'<0$ and all $r'_i$ are negative and $<r'$. Therefore, all $r_i$ are $>r$. Now the opposite inequality of~\eqref{monori} holds and the method of the preceding section fails. For~log-concave densities, however,~\eqref{monori} can be replaced by a similar inequality in the right direction.

A density $f$ is \emph{log-concave} if $\log f$ is concave in its support, i.e., for all $0<\mu<1$,
\begin{equation}\label{logconcave}
f(x)^\mu f(y)^{1-\mu} \leq  f(\mu x+ (1-\mu) y).
\end{equation}

\begin{Theorem}[Fradelizi, Madiman and Wang~\cite{FradeliziMadimanWang16}]\label{logconcaveentropyconcave}
If $X$ has a log-concave density, then $h_r(rX)-rh_r(X)= (1-r) h_r(X)+n\log r $ is concave in~$r$.
\end{Theorem}
This concavity property is used in~\cite{FradeliziMadimanWang16} to derive a sharp ``varentropy bound''.
Section~\ref{transportationvarentropy} provides an alternate transportation proof along the same lines as in Section~\ref{repi1sec}.

By Theorem~\ref{logconcaveentropyconcave}, since $n\log r + (1-r) h_r(X)$ is concave and vanishes for $r\!=\!1$, the slope $\frac{n\log r + (1-r) h_r(X)-0}{r-1}$ is nonincreasing in $r$. 
In other words, $h_r(X) + n \frac{\log r }{1-r}$ is nondecreasing.
Now since all $r_i$ are~$>r$,
\begin{equation}
h_{r_i}(X)+n \tfrac{\log r_i}{1-r_i} \geq h_r(X)+n \tfrac{\log r}{1-r} \qquad (i=1,\ldots,m). 
\end{equation}
Plugging this into~\eqref{repi1mineq}, one obtains
\begin{align}
h_r&\Bigl( \textstyle\sum\limits_{i=1}^m\sqrt{\lambda_i}X_i\Bigr) -  \textstyle\sum\limits_{i=1}^m \lambda_i h_{r}(X_i)
\notag\\&
\geq n \bigl(\tfrac{\log r}{1-r}- \textstyle\sum\limits_{i=1}^m \lambda_i\frac{\log r_i}{1-r_i}\bigr)
+ \frac{n}{2} r' \bigl( \frac{\log r}{r}-\textstyle\sum\limits_{i=1}^m\frac{\log r_i}{r_i}\bigr)\notag\\
&=\frac{n}{2} r' \bigl( \textstyle\sum\limits_{i=1}^m\frac{\log r_i}{r_i}-\frac{\log r}{r}\bigr) \label{repi2mineqlogc}
\end{align}
where we have used that $\lambda_i=r'/r'_i$ for $i=1,2,\ldots,m$. 

Notice that, quite surprisingly, the r.h.s. of~\eqref{repi2mineqlogc}
for $r<1$ ($r'<0$) is the opposite of that of~\eqref{repi2mineq} for $r>1$ ($r'>0$). However, since $r'$ is now negative, the r.h.s. is exactly equal to $\frac{n}{2}A(\lambda)$ which is still convex and negative.
For this reason, the proofs of the following theorems for $r<1$ are such repeats of the theorems obtained previously for $r>1$.

\begin{Proposition}
The Rényi EPI~\eqref{repic} for log-concave densities holds for $c=r^{-r'/r}\bigl(1-\frac{1}{mr'}\bigr)^{1-mr'}$ and $r<1$.
\end{Proposition}
\begin{proof}
Identical to that of Theorem~\ref{thmrepic} except for the change $|r'|=-r'$ in the expression of $A(\lambda)$.
\end{proof}

\begin{Proposition}[Marsiglietti and Melbourne\cite{MarsigliettiMelbourne18}]
The Rényi EPI~\eqref{repialpha} log-concave densities holds for
$\alpha=\bigl[1+{|r'| \frac{\log_2 r}{r} +(2|r'|+1) \log_2 \bigl(1+\frac{1}{2|r'|}\bigr)}  \bigr]^{-1}$ and $r<1$.
%$\alpha=\Bigl(1-\frac{r' \frac{\log r}{r} +(mr'-1) \log \bigl(1-\frac{1}{mr'}\bigr)}{\log m}  \Bigr)^{-1}$.
%and this value of $\alpha$ cannot be improved using the method of this paper by making it depend on $m$.
\end{Proposition}
\begin{proof}
Identical to that of Theorem~\ref{thmrepialpha} except for the change $|r'|=-r'$ in the expression of $A(\lambda)$.
%By Lemma~\ref{charact} for $c=1$ we only need to check that the r.h.s. of~\eqref{repi2mineqlogc} is greater than $\frac{n}{2}(1/\alpha-1)H(\lambda)$ for any choice of the $\lambda_i$'s, that is, for any choice of exponents $r_i$ such that $\sum_{i=1}^m \frac{1}{r'_i} = \frac{1}{r'}$. Thus~\eqref{repialpha} will hold for 
%$\frac{1}{\alpha}-1 = \min_\lambda \frac{A(\lambda)}{H(\lambda)}$.
\end{proof}

%\begin{Remark}
%The case $m=2$ yields $\alpha=\frac{1-r}{(r+1)\log_2(r+1)-r\log_2 r-2r}$ which was found by Marsiglietti and Melbourne~\cite{MarsigliettiMelbourne18}.  Again, the exponent of the theorem is not improved for $m>2$. 
%\end{Remark}

\begin{Proposition}
The REPI~\eqref{repig} for log-concave densities holds for $c\!=\!\bigl[mr^{-r'/r}\bigl(1-\frac{1}{mr'}\bigr)^{1-mr'}\bigr]^\alpha \!/ m$ where $r<1$, $0\!<\!\alpha\!<\!1$.
\end{Proposition}
\begin{proof}
It is identical to that of Theorem~\ref{thmrepig} except for the change $|r'|=-r'$ in the expression of $A(\lambda)$.
\end{proof}

\section{Relative and Conditional Rényi Entropies}\label{iisec}

Before turning to transportations proofs of Theorems~\ref{repi1m} and~\ref{logconcaveentropyconcave}, it is convenient to review some definitions and properties. The following notions were previously used for discrete variables, but can be easily adapted to variables with densities.

\begin{Definition}[Escort Variable~\cite{Bercher09}]\label{escort}
If $f\in L^r(\R^n)$, its {escort density} of exponent $r$  is defined by
\begin{equation}
f_r(x) = {f^r(x)}\Big/{\int_{\R^n} f^r(x) \d x}.
\end{equation}
Let $X_r\sim f_r$ denote the corresponding \emph{escort random variable}.
\end{Definition}

%We mention, in passing, the following identities.
\begin{Proposition}
Let $r\ne 1$ and assume that $X\sim f\in L^s(\R^n)$ for all $s$ in a neighborhood of $r$. Then
\begin{align}
\frac{\partial}{\partial r} \bigl( (1-r) h_r(X) \bigr) &= \E \log f(X_r) = - h(X_r\|X)\\
\frac{\partial}{\partial r} h_r(X)  &= - \frac{1}{(1-r)^2} D(X_r\|X)\leq 0\label{identitydiff}\\
\frac{\partial^2}{\partial r^2} \bigl( (1-r) h_r(X) \bigr) &= \mathrm{Var} \log f(X_r).\label{identityvar}
\end{align} 
where $h(X\|Y)=\int f \log (1/g)$ denotes cross-entropy and $D(X\|Y)=\int f \log (f/g)$ is the Kullback-Leibler divergence.
\end{Proposition}

\begin{proof} 
By the hypothesis, one can differentiate under the integral sign.
It is easily seen that
$\frac{\partial}{\partial r} \bigl( (1-r) h_r(X) \bigr)\!=\!\frac{\partial}{\partial r} \log \int f^r$ $=\! \int f_r \log f$. Taking another derivative yields $\frac{\partial}{\partial r} \frac{\int f^r \log f}{\int f^r} = \int f_r(\log f)^2 - (\int f_r \log f)^2$.
Since $\frac{\partial}{\partial r} \bigl( (1-r) h_r(X) \bigr)=(1-r)\frac{\partial}{\partial r} h_r(X) - h_r(X)$ we have $(1-r)^2 \frac{\partial}{\partial r} h_r(X)=\int f_r \log (f/f^r)+\log\int f^r =\int f_r \log (f/f_r)$.
\end{proof}

Eq.~\eqref{identitydiff} gives a new proof that $h_r(X)$ is nonincreasing in~$r$. It is strictly decreasing if $X_r$ is not distributed as $X$, that is, if $X$ is not uniformly distributed. Equation~\eqref{identityvar} shows that $(1-r)h_r(X)$ is convex in $r$, that is, $\int f^r$ is log-convex in $r$ (which is essentially equivalent to H\"older's inequality).

%\begin{Remark}
%Notice the identity established in the proof:
%\begin{equation}\label{identitydiff}
%\frac{\partial}{\partial p} h_p(X) =   - \frac{D(f_p\|f)}{(1-p)^2}.
%\end{equation}
%
%A similar formula for discrete variables can be found in~(\cite{BeckSchloegl93} \S~5.3).
%\end{Remark}

\begin{Definition}[Relative Rényi Entropy~\cite{LapidothPfister16}]
Given $X\sim f$ and $Y\sim g$, their \emph{relative Rényi entropy} of exponent $r$ (relative $r$-entropy) is given by 
$$
\Delta_r(X\|Y)=D_{\frac{1}{r}}(X_r\|Y_r)
$$
where $D_r(X\|Y)=\frac{1}{r-1}\log\int f^r g^{1-r}$ is the $r$-divergence~{\upshape\cite{vanErvenHarremoes14}}.
\end{Definition}
When $r\to 1$ both the relative $r$-entropy and the $r$-divergence tend to the Kullback-Leibler divergence $D(X\|Y)=\Delta(X\|Y)$ (also known as the relative entropy).
For $r\ne 1$ the two notions do not co\"incide. It is easily checked from the definitions that
\begin{equation}
\Delta_r(X\|Y)\!=\! -r'\log\!\!\int\!\! f_r^{1/r} g_r^{1/r'}
\!=\!-r'\!\log\E\bigl( g_r^{1/r'}\!(X)\bigr) - h_r\!(X)\label{rentropydiff}
\end{equation}
%and\vspace*{-1ex}
\begin{equation}\label{rentropy}
h_r(X)= -r'\log\E\bigl( f_r^{1/r'}(X)\bigr).
\end{equation}
Thus, just like for the case $r=1$, the relative $r$-entropy~\eqref{rentropydiff} is the difference between the expression of the $r$-entropy~\eqref{rentropy} in which $f$ is replaced by $g$, and the $r$-entropy itself.

%%We also use a conditional version of the relative $r$-entropy:
%\begin{Definition}[Conditional Relative Rényi Entropy]
%The \emph{conditional relative $r$-entropy} of $X$ and $Y$ given $Z$ is 
%$$
%\Delta_r(X\|Y|Z)=D_{\frac{1}{r}}(X_r\|Y_r|Z)
%$$
%where $D_r(X\|Y|Z)=D_r(XZ\|YZ)$ denotes the conditional Rényi $r$-divergence~{\upshape\cite{Verdu15}}.
%\end{Definition}
%Similarly as for~\eqref{rentropydiff} it is easily checked that 
%\begin{equation}\label{rentropydiffcond}
%\Delta_r(X\|Y|Z)=-r'\log\E\bigl( g_r^{1/r'}\!(X|Z)\bigr) - h_r(X|Z)
%\end{equation}
%where\vspace*{-1.5ex}
%\begin{equation}
% h_r(X|Z)=-r'\!\log\E \| f(\cdot|Z)\|_r = -r' \!\log \E f_r^{1/r'}\!(X|Z)
%\end{equation}
%is Arimoto's conditional R\'enyi entropy~\cite{FehrBerens14}.

Since the Rényi divergence $D_r(X\|Y)=\frac{1}{r-1}\int f^r g^{1-r}$ is nonnegative and vanishes if and only if the two distributions $f$ and $g$ co\"incide, 
the relative entropy $\Delta_r(X\|Y)$ enjoys the same property. From~\eqref{rentropydiff} we have the following
\begin{Proposition}[Rényi-Gibbs' inequality]\label{RGI}%
If $X\sim f$,
\begin{equation}
h_r(X)\leq  -r'\log\E\bigl( g_r^{1/r'}(X)\bigr) 
\end{equation}
 for any density $g$, with equality if and only if $f=g$ a.e.
\end{Proposition}
\noindent Letting $r\to 1$ one recovers the usual Gibbs' inequality. 

\begin{Definition}[Arimoto's Conditional R\'enyi Entropy~\cite{FehrBerens14}]
$$
 h_r(X|Z)=-r'\!\log\E \| f(\cdot|Z)\|_r = -r' \!\log \E f_r^{1/r'}\!(X|Z)
$$
\end{Definition}
Proposition~\ref{RGI} applied to $f(x|z)$ and $g(x|z)$ gives the inequality $h_r(X|{Z=z})\leq  -r'\log\E\bigl( g_r^{1/r'}(X|Z=z)\bigr)$ which, averaged over $Z$, yields the following conditional Rényi-Gibbs' inequality
\begin{equation}\label{condrgibbsineq}
h_r(X|Z)\leq  -r'\log\E\bigl( g_r^{1/r'}\!(X|Z)\bigr).
\end{equation}
If in particular we put $g(x|z)=f(x)$ independent of $z$, the r.h.s. becomes equal to~\eqref{rentropy}.  We have thus obtained a simple proof of the following
\begin{Proposition}[Conditioning reduces $r$-entropy~\cite{FehrBerens14}]
\begin{equation}
h_r(X|Z)\leq h_r(X)
\end{equation}
with equality if and only if $X$ and $Z$ are independent.
\end{Proposition}

%Similarly from~\eqref{rentropydiffcond}, we have the conditional form
%\begin{equation}\label{condrgibbsineq}
%h_r(X|Z)\leq  -r'\log\E\bigl( g_r^{1/r'}\!(X|Z)\bigr).
%\end{equation}
%If in particular we put $g(x|z)=f(x)$ independent of $z$, the r.h.s. becomes equal to~\eqref{rentropy}.  We have thus obtained a simple proof of the fact that conditioning  reduces $r$-entropy~\cite{FehrBerens14}:
%$h_r(X|Z)\leq h_r(X)$ with equality if and only if $X$ and $Z$ are independent.
%\begin{Proposition}[Conditioning reduces $r$-entropy~\cite{FehrBerens14}]
%\begin{equation}
%h_r(X|Z)\leq h_r(X)
%\end{equation}
%with equality if and only if $X$ and $Z$ are independent.
%\end{Proposition}

Another important property is the data processing inequality~\cite{vanErvenHarremoes14} which implies $D_r(T(X)\|T(Y)) \leq D_r(X\|Y)$ for any transformation $T$. The same holds for relative $r$-entropy when the transformation is applied to escort variables:
\begin{Proposition}[Data processing inequality for relative $r$-entropy]
If $X^*, Y^*, X, Y$ are random vectors such that 
\begin{equation}\label{transportescort}
X_r =  T(X^*_r) \quad \text{and}\quad Y_r =  T(Y^*_r),
\end{equation}
then $D(X\|Y)\leq D(X^*\|Y^*)$. 
\end{Proposition}
\begin{proof}
$D(X\|Y)=D_{\frac{1}{r}}(X_r\|Y_r)=D_{\frac{1}{r}}(T(X^*_r)\|T(Y^*_r))\leq D_{\frac{1}{r}}(X^*_r\|Y^*_r) =D(X^*\|Y^*)$.
\end{proof}
When $T$ is invertible, inequalities in both directions hold:
\begin{Proposition}[Relative $r$-entropy preserves transport]\label{transportpreserv}
For an invertible transport $T$ satisfying~\eqref{transportescort}, $D(X\|Y)= D(X^*\|Y^*)$.
\end{Proposition}
From~\eqref{rentropydiff} the equality $D(X\|Y)= D(X^*\|Y^*)$ can be rewritten as the following identity:
\begin{equation}\label{transportpreservationidentity}
-r'\log\E\bigl( g_r^{\frac{1}{r'}}\!(X)\bigr) - h_r(X)
\!=\!-r'\log\E\bigl( {g^*}_r^{\frac{1}{r'}}\!(X^*)\bigr) - h_r(X^*).
\end{equation}
Assuming $T$ is a diffeomorphism, the density $g^*_r$ of $Y^*_r$ is given by the change of variable formula $g^*_r(u)=g_r(T(u)) |T'(u)|$ where the Jacobian $|T'(u)|$ is the absolute value of the determinant of the Jacobian matrix $T'(u)$. In this case~\eqref{transportpreservationidentity} can be rewritten as
\begin{equation}\label{transportpreservationidentity2}
\begin{split}
-r'&\log\E\bigl( g_r^{\frac{1}{r'}}\!(X)\bigr) - h_r(X)
\\&
=-r'\log\E\bigl( {g}_r^{\frac{1}{r'}}\!(T(X^*))|T'(X^*)|^{\frac{1}{r'}}\bigr) - h_r(X^*). 
\end{split}
\end{equation}
%which is valid for any $g_r$.

\section{A Transportation Proof of Theorem~\ref{repi1m}}\label{repi1sec}

We proceed to prove~\eqref{repi1mineqGauss}. It is easily seen, using finite induction on $m$, that it suffices to prove the corresponding inequality for $m=2$ arguments:
\begin{equation}\label{repi1equiv}
\begin{split}
&h_r(\sqrt{\lambda}X\!+\! \sqrt{1\!-\!\lambda} Y) \!-\! \lambda h_p(X) \!-\! (1\!-\!\lambda) h_q(Y) \\&\;\geq
h_r(\sqrt{\lambda}X^*\!+\! \sqrt{1\!-\!\lambda} Y^*) \!-\! \lambda h_p(X^*) \!-\! (1\!-\!\lambda) h_q(Y^*)   
\end{split}
\end{equation}
with equality  if and only if $X,Y$ are i.i.d.\@ Gaussian.
Here $X^*$ and $Y^*$ are i.i.d.\@ standard Gaussian $\mathcal{N}(0,\mathbf{I})$ and the triple $(p,q,r)$ and its associated $\lambda\in(0,1)$ satisfy  the following conditions: $p,q,r$ have conjugates $p',q',r'$ of the same sign which satisfy
$
\frac{1}{p'}+\dfrac{1}{q'}=\frac{1}{r'}
$
(that is, $\frac{1}{p}+\frac{1}{q}=1+\frac{1}{r}$) and
$
\lambda= \frac{r'}{p'} = 1- \frac{r'}{q'}.
$
%We emphasize the two distinct situations:
%\begin{itemize}
%\item either $p,q,r>1$ and $p',q',r'>1$; in this case $r'<p'$, $r'<q'$, and $r>p$, $r>q$;
%\item or  $0<p,q,r<1$ and $p',q',r'<0$; in this case $|r'|<|p'|$, $|r'|<|q'|$ and 
%$r<p$, $r<q$.
%\end{itemize}

\begin{Lemma}[{Normal Transport}]\label{gt}
Let $f$ be given and $X^*\!\sim\! \mathcal{N}(0,\sigma^2\mathbf{I})$. There exists a diffeomorphism $T:\R^n\to\R^n$ with log-concave Jacobian $|T'|$ such that $X=T(X^*) \sim f$.
\end{Lemma}
\noindent Thus $T$ transports normal $X^*$ to $X$. The log-concavity property is that for any such transports $T,U$ and $\lambda\in(0,1)$, we have
\begin{equation}\label{kyfan}
|T'(X^*)|^\lambda |U'(Y^*)|^{1-\lambda} \leq |\lambda T'(X^*)+(1-\lambda)U'(Y^*)|.
\end{equation}
The proof of Lemma~\ref{gt} is very simple for one-dimensional variables~\cite{Rioul17a}, where $T$ is just an increasing function with continuous derivative $T'>0$ and where~\eqref{kyfan} is the classical arithmetic-geometric inequality. % of arithmetic and geometric means.

For dimensions $n>1$, Lemma~\ref{gt} comes into two flavors:\\
%\begin{itemize}
%\item 
\emph{(i) Kn\"othe maps:} $T$ can be chosen such that its Jacobian matrix $T'$ is (lower) triangular with positive diagonal elements (Kn\"othe--Rosenblatt map~\cite{Rosenblatt52,Knothe57}). Two different elementary proofs are given in~\cite{Rioul17}. Inequality~\eqref{kyfan} results from the concavity of the logarithm applied to the Jacobian matrices' diagonal elements.
%\item 
\\\emph{(ii) Brenier maps:}
$T$ can be chosen such that its Jacobian matrix $T'$ is symmetric positive definite (Brenier map~\cite{Brenier91,McCann95}). In this case~\eqref{kyfan} is Ky Fan's inequality~\cite[\S~17.9]{CoverThomas06}.
%\end{itemize}

\medskip

The key argument is now the following. Considering escort variables, by transport (Lemma~\ref{gt}), one can write $X_p = T(X^*_p)$ and $Y_q = U(Y^*_q)$
%\begin{align}
%X_p &= T(X^*_p)\\
%Y_q &= U(Y^*_q)
%\end{align}
for two diffeomorphims $T$ and $U$ satisfying~\eqref{kyfan}.
%, where $X^*,Y^*$ are, say, i.i.d.\@ standard Gaussian $\mathcal{N}(0,\mathbf{I})$. 
%(It follows that $X^*_p\sim \mathcal{N}(0,\mathbf{I}/p)$ and $Y^*_q\sim \mathcal{N}(0,\mathbf{I}/q)$.)
Then by transport preservation (Proposition~\ref{transportpreserv}), we have $\lambda\Delta_p(X\|U)+(1-\lambda)\Delta_p(Y\|V)=\lambda\Delta_p(X^*\|U^*)+(1-\lambda)\Delta_p(Y^*\|V^*)$ for any $U\sim \phi$ and $V\sim\psi$, which from~\eqref{transportpreservationidentity2} can be easily rewritten in the form
\begin{multline}\label{someidentity}
-r'\log\E\bigl( \chi^\frac{1}{r'}(X,Y)\bigr) - \lambda h_p(X)-(1-\lambda)h_q(Y)
\\=-r'\log\E\Bigl(\!\bigl( \chi(T(X^*),U(Y^*))
|T'(X^*)|^\lambda |U'(Y^*)|^{1-\lambda}\bigr)^\frac{1}{r'}\!\Bigr)\\
  - \lambda h_p(X^*)-(1-\lambda) h_q(Y^*)
\end{multline}
where we have noted $\chi(x,y)=\phi_p^{\lambda}(x)\psi_q^{1-\lambda}(y)$.
Such an identity holds, by the change of variable $x=T(x^*),y=U(y^*)$, for any function $\chi(x,y)$ of $x$ and $y$.
Now from~\eqref{rentropy} we have
\begin{equation*}
h_r(\sqrt{\lambda}X\!+\! \sqrt{1\!-\!\lambda} Y) = -r'\log\E\bigl( \theta_r^{1/r'}\!\!(\sqrt{\lambda}X\!+\! \sqrt{1\!-\!\lambda} Y)\bigr).
\end{equation*}
where $\theta$ is the density of $\sqrt{\lambda}X\!+\! \sqrt{1\!-\!\lambda} Y$.
Therefore, the l.h.s. of~\eqref{repi1equiv} can be written as
\begin{align}
&h_r(\sqrt{\lambda}X\!+\! \sqrt{1\!-\!\lambda} Y) \!-\! \lambda h_p(X) \!-\! (1\!-\!\lambda) h_q(Y) \\
&=\!-\!r'\log\E\bigl( \theta_r^{\frac{1}{r'}}\!\!(\sqrt{\lambda}X\!+\! \sqrt{1\!-\!\lambda} Y)\bigr) \!-\! \lambda h_p(X) \!-\! (1\!-\!\lambda) h_q(Y)  \notag
\end{align}
Applying~\eqref{someidentity} to $\chi(x,y)=\theta_r(\sqrt{\lambda}x\!+\! \sqrt{1\!-\!\lambda} y)$ and using the inequality~\eqref{kyfan} gives 
\begin{align}\label{ineqqq}
h_r&(\sqrt{\lambda}X\!+\! \sqrt{1\!-\!\lambda} Y) \!-\! \lambda h_p(X) \!-\! (1\!-\!\lambda) h_q(Y)\\
 &\geq -r'\log\E\bigl( 
 \phi^\frac{1}{r'}(X^*,Y^*)
 \bigr) \!-\! \lambda h_p(X^*) \!-\! (1\!-\!\lambda) h_q(Y^*)\notag
\end{align}
where $\phi(x^*,y^*)=\theta_r(\sqrt{\lambda}T(x^*)\!+\! \sqrt{1\!-\!\lambda} U(y^*))\cdot
|\lambda T'(x^*)\!+\!(1\!-\!\lambda)U'(y^*)|$.
To conclude we need the following
\begin{Lemma}[{Normal} Rotation~\cite{Rioul17}]\label{gr}
If $X^*,Y^*$ are i.i.d. Gaussian, then~for any $0<\lambda<1$, the rotation\pagebreak[1]
\begin{equation}
%\begin{cases}
\tilde{X}\!=\!\sqrt{\lambda}\;X^* \!+\! \sqrt{1-\lambda}\;Y^*,
%\\
\quad
\tilde{Y}\!=\!-\sqrt{1-\lambda}\;X^* \!+\! \sqrt{\lambda}\;Y^*
%\end{cases} 
\end{equation}
%\begin{equation}
%\begin{pmatrix}
%\tilde{X} \\\tilde{Y}
%\end{pmatrix}
% =
%\begin{pmatrix}
%\sqrt{\lambda} & \sqrt{1-\lambda}\\
%-\sqrt{1-\lambda} & \sqrt{\lambda}\\
%\end{pmatrix}
%\begin{pmatrix}
%X^* \\ Y^*
%\end{pmatrix}
%\end{equation}
yields i.i.d.\@ Gaussian variables $\tilde{X},\tilde{Y}$. 
\end{Lemma}
Lemma~\ref{gr} is easy proved considering covariance matrices. 
A deeper result (Bernstein's lemma, not used here) states that this property of remaining i.i.d.\@ by rotation characterizes the Gaussian distribution~\cite[Lemma~4]{Rioul17a}~\cite[Chap.~5]{Bryc95}). 
%This explains why one obtains equality in the EPI only for Gaussian variables (see~\cite{Rioul17a} for more details).
%\pagebreak%%%%% HACK!!!

Since the starred variables can be expressed in terms of the tilde variables by the inverse~rotation
%\begin{equation}\label{invrot}
% \begin{array}{lcl}
$X^*\!=\!\sqrt{\lambda}\;\tilde{X}  - \sqrt{1-\lambda}\;\tilde{Y}$,
%\\
%\quad
$Y^*\!=\!\sqrt{1-\lambda}\;\tilde{X} + \sqrt{\lambda}\;\tilde{Y}$,
%\end{array} 
%\end{equation}
%\begin{equation}\label{invrot}
%\begin{pmatrix}
%X^* \\ Y^*
%\end{pmatrix}
% =
%\begin{pmatrix}
%\sqrt{\lambda} & -\sqrt{1-\lambda}\\
%\sqrt{1-\lambda} & \sqrt{\lambda}\\
%\end{pmatrix}
%\begin{pmatrix}
%\tilde{X} \\\tilde{Y}
%\end{pmatrix}.
%\end{equation}
inequality~\eqref{ineqqq} can be written as
\begin{align}\label{afterrotation}
&h_r(\sqrt{\lambda}X+\sqrt{1-\lambda}Y) - \lambda h_p(X) -(1-\lambda)h_q(Y)\\
&\;\geq -r' \log \E\bigl(   \psi^{1/r'}(\tilde{X}|\tilde{Y})   \bigr)- \lambda h_p(X^*) -(1-\lambda)h_q(Y^*),\notag
\end{align}
where 
%\phi(\sqrt{\lambda}\tilde{X}  - \sqrt{1-\lambda}\tilde{Y},\sqrt{1-\lambda}\tilde{X} + \sqrt{\lambda}\tilde{Y})=
$
\psi(\tilde{x}|\tilde{y}) =
\theta_r(\sqrt{\lambda}T(\sqrt{\lambda}\tilde{x}  \!-\! \sqrt{1\!-\!\lambda}\tilde{y})\!+\! \sqrt{1\!\!-\!\!\lambda} U(\sqrt{1\!-\!\lambda}\tilde{x} + \sqrt{\lambda}\tilde{y}))\cdot
|\lambda T'(\sqrt{\lambda}\tilde{x}  \!-\! \sqrt{1\!-\!\lambda}\tilde{y})\!+\!(1\!-\!\lambda)U'(\sqrt{1\!-\!\lambda}\tilde{x} + \sqrt{\lambda}\tilde{y})|$.
Making the change of variable $z=\sqrt{\lambda}T(\sqrt{\lambda}\tilde{x}-\sqrt{1-\lambda}\tilde{y})+\sqrt{1-\lambda}U(\sqrt{1-\lambda}\tilde{x}+\sqrt{\lambda}\tilde{y})$, we check that
$\int\!\psi(\tilde{x}|\tilde{y}) \d \tilde{x} = \int\!\theta_r(z)\d z =1$
since $\theta_r$ is a density. Hence, $\psi(\tilde{x}|\tilde{y})$ is a conditional density, and by~\eqref{condrgibbsineq},
\begin{equation}\label{condpsi}
-r' \log \E\bigl(   \psi^{1/r'}(\tilde{X}|\tilde{Y})   \bigr) \geq h_r(\tilde{X}|\tilde{Y}) 
\end{equation}
where $h_r(\tilde{X}|\tilde{Y}) = h_r(\tilde{X})=h_r(\sqrt{\lambda}\;X^* + \sqrt{1-\lambda}\;Y^*)$ since $\tilde{X}$ and $\tilde{Y}$ are independent.
Combining with~\eqref{afterrotation} yields the announced inequality~\eqref{repi1equiv}.

It remains to settle the equality case in~\eqref{repi1equiv}. From the above proof, equality holds in~\eqref{repi1equiv} if and only if both~\eqref{kyfan} and~\eqref{condpsi} are equalities. The rest of the argument depends on whether Kn\"othe or Brenier maps are used:\\
\emph{(i) Kn\"othe maps:}
In the case of Kn\"othe maps, Jacobian matrices are triangular and equality in~\eqref{kyfan} holds if and only if for all $i=1,2,\ldots,n$,
$\frac{\partial T_i}{\partial x_i}(X^*) =  \frac{\partial U_i}{\partial y_i}(Y^*) \text{ a.s.}$
Since $X^*$ and $Y^*$ are independent Gaussian variables, this implies that $\frac{\partial T}{\partial x_i}$ and $\frac{\partial U}{\partial y_i}$ are constant and equal. In particular the Jacobian $|\lambda T'(\sqrt{\lambda}\tilde{x}-\sqrt{1-\lambda}\tilde{y}) + (1-\lambda) U'(\sqrt{1-\lambda}\tilde{x}+\sqrt{\lambda}\tilde{y})| $ is~constant. 
Now since $h_r(\tilde{X}|\tilde{Y}) = h_r(\tilde{X})$ equality in~\eqref{condpsi} holds only if $ \psi(\tilde{x}|\tilde{y})$ does not depend on $\tilde{y}$, which implies that $\sqrt{\lambda}T(\sqrt{\lambda}\tilde{x}-\sqrt{1-\lambda}\tilde{y})+\sqrt{1-\lambda}U(\sqrt{1-\lambda}\tilde{x}+\sqrt{\lambda}\tilde{y})$ does not depend on the value of $\tilde{y}$. Taking~derivatives with respect to $y_j$ for all $j=1,2,\ldots,n$, we have
$
  -\sqrt{\lambda}\sqrt{1-\lambda} \frac{\partial T_i}{\partial x_j}(\sqrt{\lambda}\tilde{X}-\sqrt{1-\lambda}\tilde{Y}) + \sqrt{\lambda}\sqrt{1-\lambda} \frac{\partial U_i}{\partial x_j}(\sqrt{1-\lambda}\tilde{X}+\sqrt{\lambda}\tilde{Y})=0
$
which implies 
$\frac{\partial T_i}{\partial x_j}(X^*) =  \frac{\partial U_i}{\partial y_j}(Y^*)$ a.s.
for all $i,j=1,2,\ldots,n$. In other words, $T'(X^*)=U'(Y^*)$ a.s.
\\
\emph{(ii) Brenier maps:}
In the case of Brenier maps the argument is simpler. Jacobian matrices are symmetric positive definite and by strict concavity, Ky Fan's inequality~\eqref{kyfan} is an equality only if $T'(X^*)=U'(Y^*)$ a.s.

\medskip

In both cases, since $X^*$ and $Y^*$ are independent, this implies that $T'(X^*)=U'(Y^*)$ is constant.
Therefore, $T$ and $U$ are linear transformations, equal up to an additive constant ($=0$ since the random vectors are assumed of zero mean). It follows that $X_p = T(X^*_p)$ and $Y_q = U(Y^*_q)$ are Gaussian with respective distributions $X_p\sim \mathcal{N}(0,\mathbf{K}/p)$ and $Y_q\sim \mathcal{N}(0,\mathbf{K}/q)$. Hence, $X$ and $Y$ are i.i.d.\@ Gaussian $\mathcal{N}(0,\mathbf{K})$.
This ends the proof of Theorem~\ref{repi1m}.\hfill\qedsymbol

We note that this section has provided an information-theoretic proof the strengthened Young's convolutional inequality (with optimal constants), since~\eqref{repi1equiv} is a rewriting of this convolutional inequality~\cite{DemboCoverThomas91}.

\section{A Transportation Proof of Theorem~\ref{logconcaveentropyconcave}}
\label{transportationvarentropy}

Define $r=\lambda p + (1-\lambda) q$ where $0<\lambda<1$. 
It is required to show that $(1-r) h_r(X)+n\log r \geq \lambda\bigl((1-p) h_p(X)+n\log p\bigr)+(1-\lambda)\bigl((1-q) h_q(X)+n\log q\bigr)$.

By Lemma~\ref{gt} there exists two diffeomorphisms $T,U$ such that one can write $pX_p=T(X^*)$ and $qX_q=U(X^*)$. Then, by these changes of variables $X^*$ has density
\begin{equation}
\tfrac{1}{p^n} f_p\bigl(\tfrac{T(x^*)}{p}\bigr) |T'(x^*)| = 
\tfrac{1}{q^n} f_q\bigl(\tfrac{U(x^*)}{q}\bigr) |U'(x^*)| 
\end{equation}
which can be written
$$
\frac{ f^p\bigl(\tfrac{T(x^*)}{p}\bigr) |T'(x^*)|}{\exp\bigl((1-p) h_p(X)+n\log p\bigr)} = 
\frac{f^q\bigl(\tfrac{U(x^*)}{q}\bigr) |U'(x^*)|}{\exp\bigl((1-q) h_q(X)+n\log q\bigr)}  
$$
Taking the geometric mean, integrating over $x^*$ and taking the logarithm gives the representation
\begin{multline*}
  \lambda\bigl((1-p) h_p(X)+n\log p\bigr)+(1-\lambda)\bigl((1-q) h_q(X)+n\log q\bigr)
  \\=\! \log \!\int\!   f^{\lambda p}\bigl(\tfrac{T(x^*)}{p}\bigr) 
                       f^{(1\!-\!\lambda)q}\bigl(\tfrac{U(x^*)}{q}\bigr) 
                       |T'(x^*)|^\lambda|U'(x^*)|^{1\!-\!\lambda}\d x^*\!.
\end{multline*}
Now, by log-concavity~\eqref{logconcave} (with $\mu=\lambda p/r$) and~\eqref{kyfan},
\begin{align*}
  &\lambda\bigl((1\!-\!p) h_p(X)+n\log p\bigr)+(1\!-\!\lambda)\bigl((1\!-\!q) h_q(X)+n\log q\bigr)\notag\\
  &\leq \log \smash{\int}   f^r\bigl(\tfrac{\lambda T(x^*)\!+\!(1\!-\!\lambda)U(x^*)}{r}\bigr ) 
                       |\lambda T'(x^*)\!+\!(1\!-\!\lambda)U'(x^*)|\d x^*\\%\notag\\
   &= \log \bigl(r^n\!\int\!\! f^r\bigr) =      (1-r) h_r(X)+n\log r .                
\end{align*}
%by the change of variable $\tfrac{\lambda T(x^*)+(1-\lambda)U(x^*)}{r}$.
This ends the proof of Theorem~\ref{logconcaveentropyconcave}.\hfill\qedsymbol

This theorem asserts that the second derivative $\frac{\partial^2}{\partial r^2} \bigl(  (1-r) h_r(X)+n\log r \bigr)\leq 0$. From~\eqref{identityvar}
this gives $\mathrm{Var} \log f(X_r) \leq n/r^2$, that is, $\mathrm{Var} \log f_r(X_r) \leq n$. Setting $r=1$, this is the varentropy bound  $\mathrm{Var} \log f(X) \leq n$ of~\cite{FradeliziMadimanWang16}.

%%%%%%
%% Appendix:
%% If needed a single appendix is created by
%%
%\appendix
%%
%% If several appendices are needed, then the command
%%
% \appendices
%%
%% in combination with further \section-commands can be used.
%%%%%%

%\section*{Acknowledgment}

%%%%%%
%% To balance the columns at the last page of the paper use this
%% command:
%%
%\enlargethispage{-1.2cm} 
%%
%% If the balancing should occur in the middle of the references, use
%% the following trigger:
%%
%\IEEEtriggeratref{10}
%%
%% which triggers a \newpage (i.e., new column) just before the given
%% reference number. Note that you need to adapt this if you modify
%% the paper.  The "triggered" command can be changed if desired:
%%
%\IEEEtriggercmd{\enlargethispage{-20cm}}
%%
%%%%%%

%%%%%%
%% References:
%% We recommend the usage of BibTeX:
%%
%\bibliographystyle{IEEEtran}
%\bibliography{renyi}

\vspace*{0.001pt}
%%
%% where we here have assume the existence of the files
%% definitions.bib and bibliofile.bib.
%% BibTeX documentation can be obtained at:
%% http://www.ctan.org/tex-archive/biblio/bibtex/contrib/doc/
%%%%%%

%% Or you use manual references (pay attention to consistency and the
%% formatting style!):
%\begin{thebibliography}{9}
%
%\bibitem{key}
%...
%
%\end{thebibliography}

\end{document}